\documentclass[pdflatex,sn-mathphys-num]{sn-jnl}

\usepackage{graphicx}%
\usepackage{multirow}
\usepackage{amsmath,amssymb,amsfonts}%
\usepackage{amsthm}%
\usepackage{mathrsfs}%
\usepackage[title]{appendix}%
\usepackage{xcolor}%
\usepackage{textcomp}%
\usepackage{manyfoot}%
\usepackage{booktabs}%
\usepackage{algorithm}%
\usepackage{algorithmicx}%
\usepackage{algpseudocode}%
\usepackage{listings}
\floatstyle{ruled}
\newfloat{listing}{tb}{lst}{}
\floatname{listing}{Listing}

\usepackage{xcolor}

\usepackage{booktabs}     
\usepackage{array}        

\usepackage{makecell}

\theoremstyle{thmstyleone}%
\newtheorem{theorem}{Theorem}
\theoremstyle{thmstyletwo}%

\theoremstyle{thmstylethree}%

\newtheorem{corollary}{Corollary}

\begin{document}

\title[VoteGCL]{Enhancing Graph-based Recommendations with \\ Majority-Voting LLM-Rerank Augmentation}


\author[1]{\fnm{Minh-Anh} \sur{Nguyen}}\email{minh.na2@vinuni.edu.vn}

\author[2]{\fnm{Bao} \sur{Nguyen}}\email{nguyenngocbaocmt@gmail.com}

\author[1]{\fnm{Ha Lan} \sur{N.T.}}\email{lan.nth@vinuni.edu.vn}

\author[3]{\fnm{Tuan Anh} \sur{Hoang}}\email{anh.hoang62@rmit.edu.vn}

\author[4]{\fnm{Duc-Trong} \sur{Le}}\email{trongld@vnu.edu.vn}

\author*[1]{\fnm{Dung D.} \sur{Le}}\email{dung.ld@vinuni.edu.vn}

\affil*[1]{\orgname{Center for AI Research, VinUniversity}, \orgaddress{\city{Hanoi}, \country{Vietnam}}}

\affil[2]{\orgname{The Chinese University of Hong Kong}, \orgaddress{\city{Hong Kong}, \country{China}}}

\affil[3]{\orgname{RMIT University}, \orgaddress{\city{Hanoi}, \country{Vietnam}}}

\affil[4]{\orgname{VNU University of Engineering and Technology}, \orgaddress{\city{Hanoi}, \country{Vietnam}}}


\abstract{Recommendation systems often suffer from data sparsity, caused by limited user-item interactions, which degrades their performance and amplifies popularity bias in real-world scenarios. This paper proposes a novel data augmentation framework that leverages Large Language Models (LLMs) and item textual descriptions to enrich interaction data. By few-shot prompting LLMs multiple times to rerank items and aggregating the results via majority voting, we generate high-confidence synthetic user-item interactions, supported by theoretical guarantees based on the concentration of measure. To effectively leverage the augmented data in the context of a graph recommendation system, we integrate it into a graph contrastive learning framework to mitigate distributional shift and alleviate popularity bias. Extensive experiments show that our method improves accuracy and reduces popularity bias, outperforming strong baselines.}




\maketitle

\section{Introduction}
Recommendation systems (RS), particularly graph-based recommendation models, often suffer from data sparsity, which significantly degrades performance and amplifies popularity bias \cite{idrissi2020systematic}. Traditional data augmentation techniques such as random walk sampling, node and edge dropout \cite{wu2021self}, and subgraph sampling \cite{fan2023graph, li2023graph} generate alternative views of the interaction graph for use in self-supervised learning. These methods are simple and can improve accuracy in low-interaction regimes. Still, they rely solely on interaction data and overlook the potential of other modalities, such as item textual descriptions. 
\begin{figure}[t]
\centering
\begin{minipage}{0.50\columnwidth}
    \includegraphics[width=0.9\textwidth]{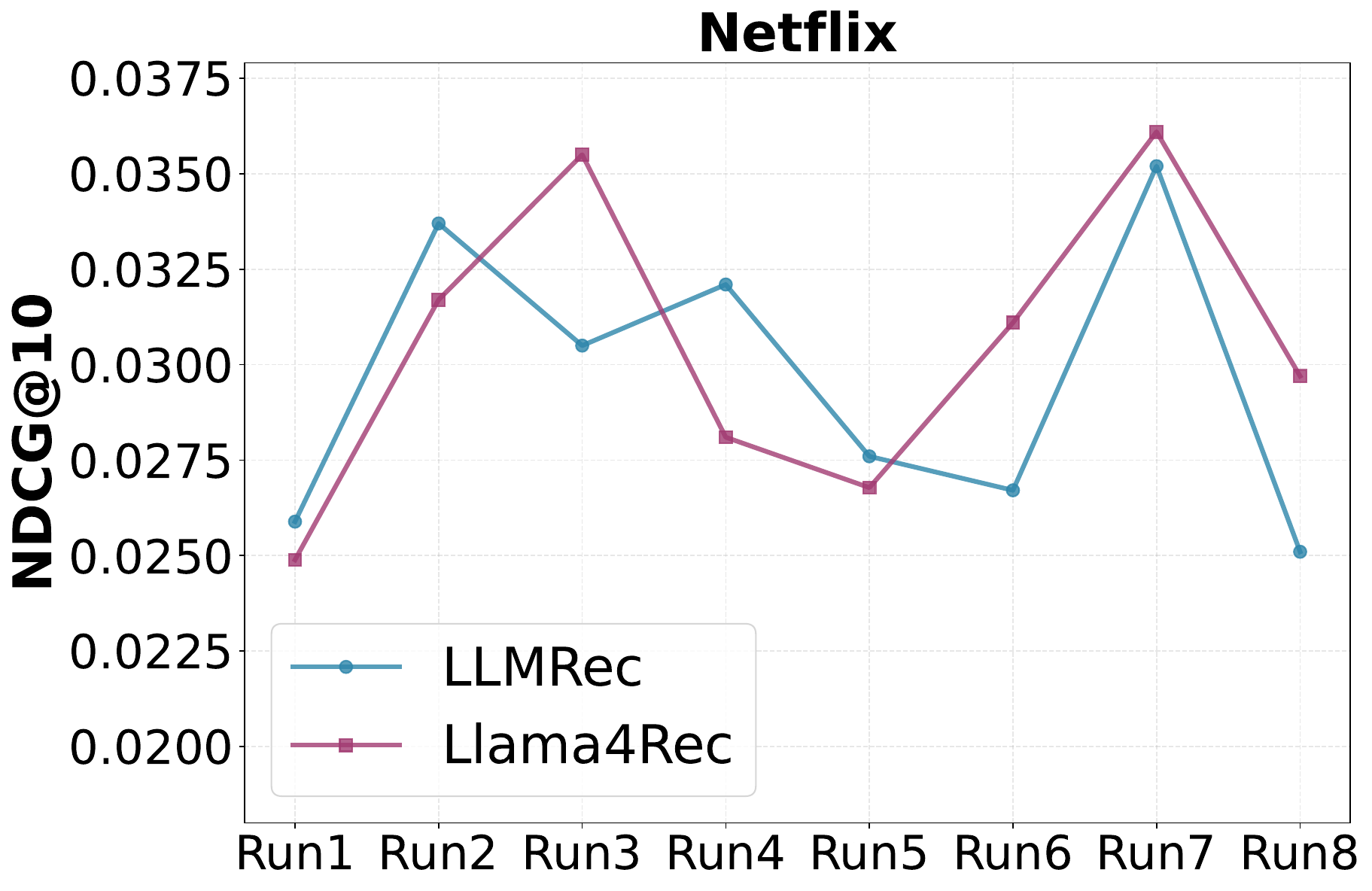}
\end{minipage}
\begin{minipage}{0.49\columnwidth}
    \centering
    \includegraphics[width=\textwidth]{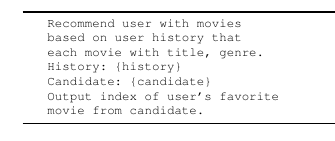}
\end{minipage}
\caption{NDCG@10 variation across repeated LLM augmentation runs on 1,000 Netflix users using a casual prompt. Variability in augmented items across runs leads to unstable recommendation quality. Higher NDCG@10 indicates better alignment with user preferences.}
\label{fig:accuracy_LLMRec}  
\end{figure}
\begin{figure}[t]
\centering
    \includegraphics[width=\columnwidth]{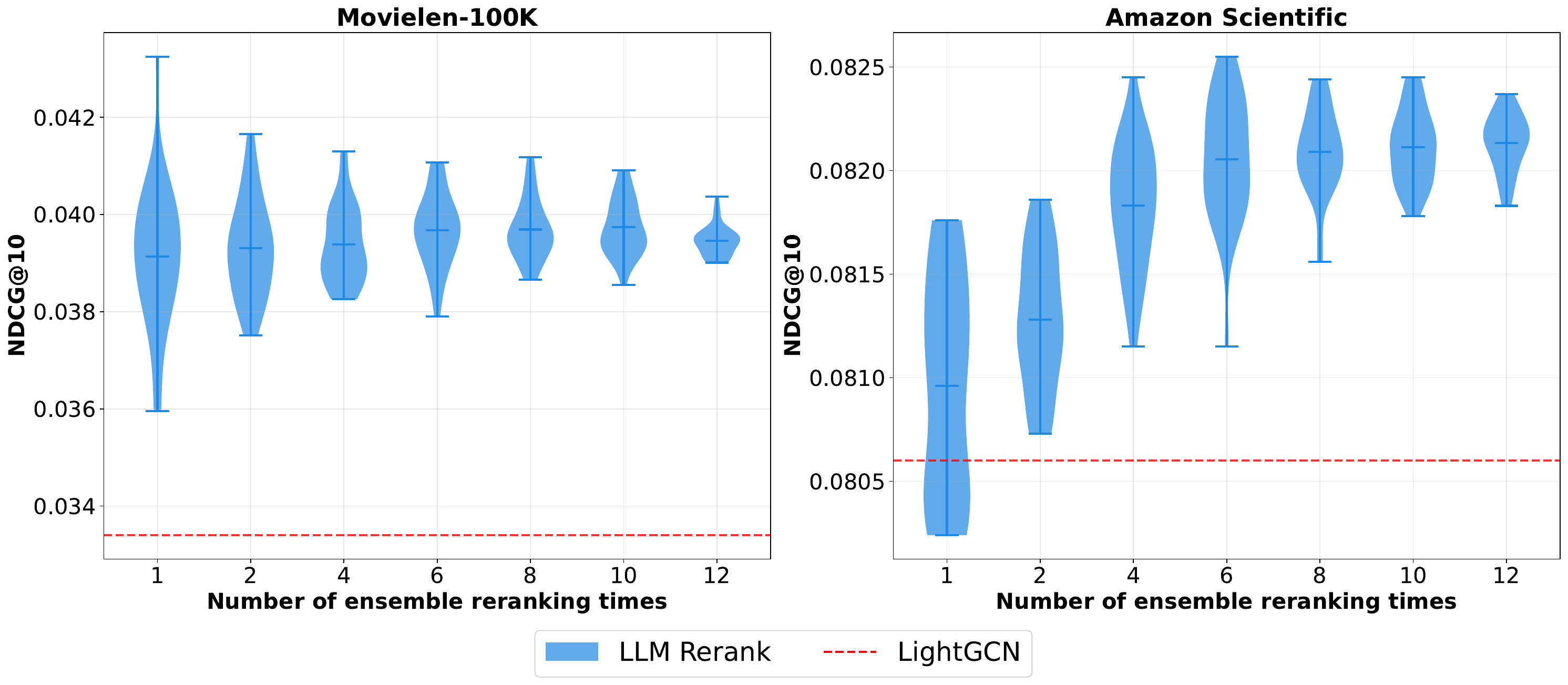}
\caption{Distribution of NDCG@10 across different reranking times over 30 repeated runs. Increasing the ensemble size via more reranking times improves the mean and reduces variance, indicating greater stability and robustness. LLM-based reranking consistently outperforms LightGCN, and ensemble methods further improve accuracy.}
\label{violin}  
\end{figure}

Recent advances leverage the reasoning capabilities of Large Language Models (LLMs) and the rich semantic information in item descriptions to enhance user and item profiling for data augmentation \cite{wu2024survey}. While LLM-based methods show promise in enriching recommendation datasets using external knowledge, they suffer from two main limitations. First, the augmentation results generated by prompting LLMs are often inconsistent across different inference runs, leading to unstable performance. Specifically, using a commonly adopted prompt template from prior works to augment interaction data, we observe that both LLMRec \cite{hou2024large} and Llama4Rec \cite{luo2024integrating} yield fluctuating NDCG@10 scores on the Netflix dataset, as shown in Figure~\ref{fig:accuracy_LLMRec}. Second, embeddings derived from the textual content of LLM-embedded models tend to be out-of-distribution compared to collaborative embeddings, making it difficult to align them in a model-agnostic framework across various RS \cite{ren2024representation, qiao2024llm4sbr}.

To address the above challenges, we reformulate LLM-based data augmentation as an item reranking task, combining few-shot prompting with majority-vote ensembling \cite{nguyenreasoning} via Reciprocal Rank Fusion (RRF) \citep{cormack2009reciprocal}. This formulation enables the LLM to better leverage its parametric knowledge of user preferences, assigning high positive scores to relevant items \cite{wu2024coral, hou2024large2}. Ensemble inference improves robustness and consistency by aggregating signals across reranking runs, as shown in Figures ~\ref{violin}. To support this ensembling approach, we propose a theoretically grounded augmentation strategy, using concentration of measure theory to derive bounds on augmentation accuracy. Rather than relying on LLM-generated embeddings, we directly utilize the augmented data in a graph contrastive learning (CL) framework. This approach aligns representations between original and augmented graphs, mitigating distributional shift and popularity bias \cite{zimmermann2021contrastive}, without modifying the base model architecture. We summarize our contributions as:
\begin{itemize}
    \item We propose \textbf{VoteGCL}, a model-agnostic, end-to-end framework that employs LLMs with few-shot in-context reranking. VoteGCL uses majority-vote reranking over candidate items to perform data augmentation and further integrates graph contrastive learning to effectively leverage the augmented data.
    \item We provide a theoretical analysis based on the concentration of measure, demonstrating that majority-vote reranking effectively extracts the LLM’s most reliable knowledge, thereby enhancing the quality of augmented interactions.
    \item Extensive experiments show that VoteGCL consistently outperforms state-of-the-art graph-based and LLM-enhanced recommender models across multiple benchmarks, improving recommendation accuracy while mitigating popularity bias. The implementation is available at: \url{https://github.com/minhkks/VoteGCL.git}.
\end{itemize}
\section{Related work}
\textbf{Graph-Augmented Recommendation.}
Graph-based RS represents user-item interactions as bipartite graphs, where nodes are users and items, and edges denote interactions. Early models used random walks to uncover latent relationships via transition probabilities~\cite{nikolakopoulos2019recwalk}. The introduction of Graph Convolutional Networks (GCNs) enabled embedding propagation better to capture collaborative signals~\cite{he2020lightgcn, wang2019neural}. Various graph augmentation techniques have been developed to address data sparsity and over-smoothing issues in graph RS. Traditional graph augmentation methods~\cite{ding2022data} can be categorized into three approaches: (1) Edge/Node Dropout, which removes redundant connections to improve robustness~\cite{wu2021self,zhou2023selfcf}; (2) Graph Diffusion, which enriches connectivity by adding new edges based on structural properties~\cite{fan2023graphda,yang2021enhanced,li2025mask,li2026leveraging}; and (3) Subgraph Sampling, which generates localized views for diverse training samples~\cite{zhang2024graph,xie2022contrastive}. While effective, these methods focus solely on graph structure and overlook valuable side information such as item content.

\textbf{LLMs-Enhanced Graph Recommendation.}
Recent advances explore integrating large language models (LLMs) into RS, typically along three lines: (1) Knowledge Enhancement, which uses LLMs’ reasoning capabilities to generate comprehensive user profiles and item descriptions by summarizing user behavior patterns and extracting rich semantic features from item content \cite{liu2024once, xi2024towards, wang2024llmrg}; (2) Model Enhancement, where LLM-generated textual embeddings of item descriptions and user profiles are integrated with collaborative signals through attention mechanisms or multi-modal alignment \cite{ren2024representation, qiao2024llm4sbr}; and (3) Interaction Enhancement employs LLMs to augment the user-item interaction space by generating synthetic interactions or directly scoring candidate items. This includes predicting missing ratings, generating counterfactual interactions to address data sparsity \cite{wei2024llmrec, song2024large}. Despite their promise, LLM-based methods face challenges such as representation misalignment and inconsistent inference due to the stochastic nature of generation.
\section{Preliminaries}
\begin{figure*}[t]  
\centering
\includegraphics[width=\columnwidth]{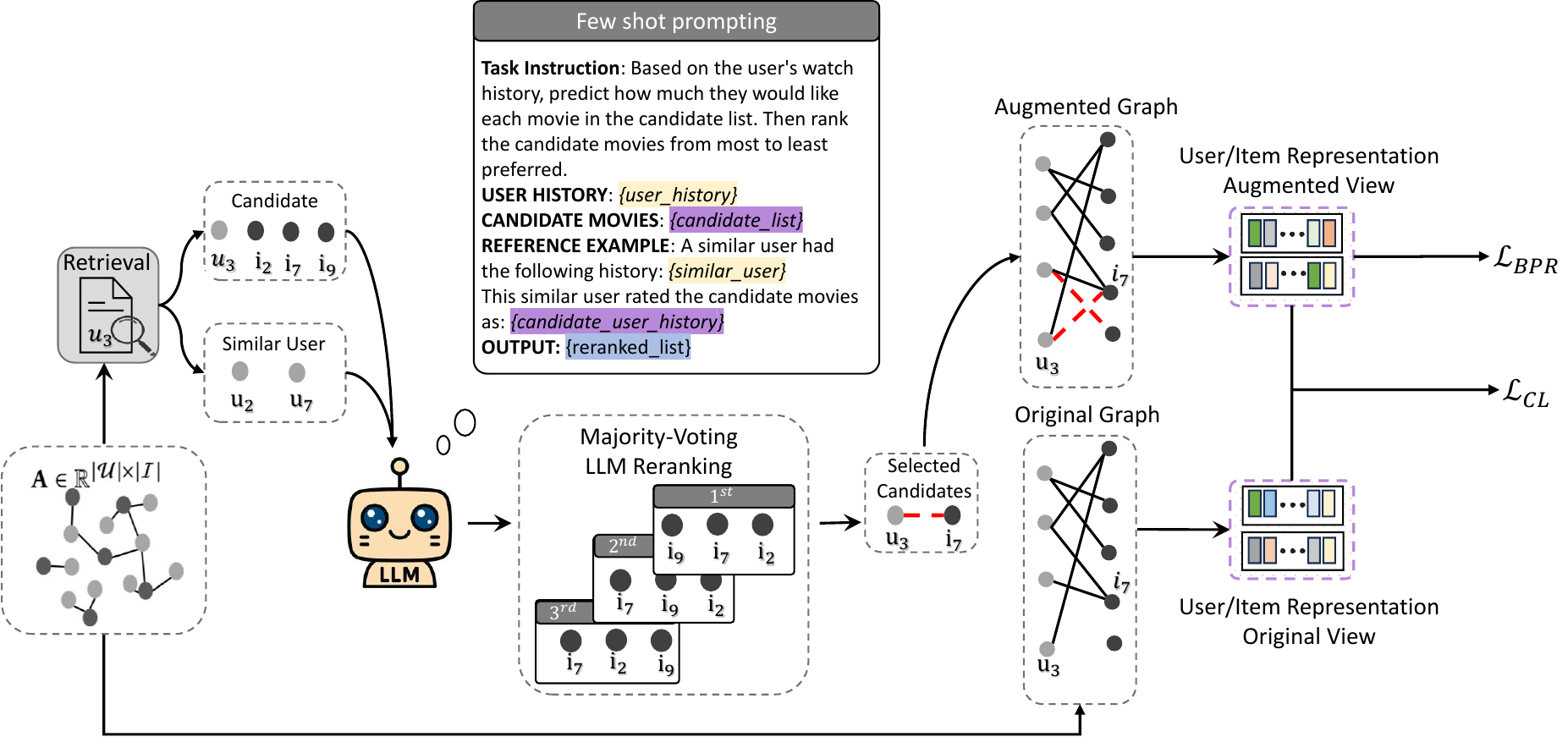}
\caption{\textbf{Overview of the VoteGCL}. The framework starts with a retrieval model that generates candidate items and identifies similar users for low-degree users using collaborative signals. An LLM reranks these candidates, and majority voting selects top-K high-confidence interactions for augmentation. The augmented data is then integrated into a graph contrastive learning framework, where two contrastive views are created via separate graph convolutions on a shared embedding matrix.}
\label{fig1}
\end{figure*}
\textbf{Problem Formulation.}
In graph-based recommendation systems, user-item interactions are modeled as a bipartite graph $\mathcal{G} = (\mathcal{V}, \mathcal{E})$, where $\mathcal{V} = \mathcal{U} \cup \mathcal{I}$ consists of user nodes $\mathcal{U}$ and item nodes $\mathcal{I}$, and $\mathcal{E}$ denotes the set of observed interactions. Each user $u \in \mathcal{U}$ and item $i \in \mathcal{I}$ is assigned a learnable embedding $\mathbf{e}_u, \mathbf{e}_i \in \mathbb{R}^d$, initialized in an embedding matrix $\mathbf{E}^{(0)} \in \mathbb{R}^{|\mathcal{V}| \times d}$ and updated during training. The model is optimized on mini-batches $\mathcal{B}$ sampled from $\mathcal{E}$. After training, we denote the resulting embedding matrix as $\mathbf{E}$, which is learned such that each user embedding $\mathbf{e}_u$ is close to the embeddings $\mathbf{e}_i$ of the items the user prefers. \\
\textbf{Vanilla Graph-based Retrieval.} We adopt LightGCN \cite{he2020lightgcn}, a simplified yet effective GCN, as the primary retrieval model to propagate information over the user-item graph $\mathcal{G}$ by updating node representations through neighborhood aggregation and combination. At the $l$-th iteration, the node representations are updated as: $\mathbf{E}^{(l)} = H(\mathbf{E}^{(l-1)}, \mathcal{G}),$ 
where $\mathbf{E}^{(l)}$ denotes the node representations at iteration $l$ which corresponds to the $l$-th layer of graph convolution, and $H$ is the neighborhood aggregation function. The embedding propagation process is mathematically formulated as a weighted combination of representations from different iterations. The final embeddings are computed as:
\begin{equation}\label{gcn}
\mathbf{E} = \frac{1}{1 + L} \left(\mathbf{E}^{(0)} + \tilde{\mathbf{A}}\mathbf{E}^{(0)} + \ldots + \tilde{\mathbf{A}}^L\mathbf{E}^{(0)}\right),
\end{equation}
where $\mathbf{E}$ denotes the final representations used for downstream prediction tasks, $L$ corresponds to the number of propagation layers, and $\tilde{\mathbf{A}} \in \mathbf{R}^{|\mathcal{V}| \times |\mathcal{V}|}$ is the normalized undirected adjacency matrix without self-connections. This iterative process allows the model to capture higher-order connectivity patterns in the graph. A common approach is to train the graph-based recommendation model using the Bayesian Personalized Ranking (BPR) \cite{rendle2012bpr} loss. The model learns to assign higher scores to observed interactions than to unobserved ones. The predicted score for user $u$ and item $i$ is computed via inner product: $\hat{y}_{ui} = \mathbf{e}_u^\top \mathbf{e}_i.$ The BPR loss is defined as:
\begin{equation}
    \mathcal{L}_{\text{BPR}} = \sum_{(u, i, j) \in \mathcal{O}} - \log \sigma\left( \hat{y}_{ui} - \hat{y}_{uj} \right),
\end{equation}
where $O=\left\{(u, i, j) \mid(u, i) \in \mathcal{E},(u, j) \in \mathcal{R}^{-}\right\}$ denotes the pairwise training data, $\mathcal{R}^{-}$is the unobserved interactions and $\sigma(\cdot)$ is the sigmoid function. In inference, items closest to $\mathbf{e}_u$ in embedding space are recommended.
\section{Methodology}
We propose a framework \textbf{VoteGCL}, as illustrated in Figure \ref{fig1}. We use a few-shot prompting LLM for inference via majority-voted reranking to construct the augmented dataset as detailed in Section~\ref{LLM rerank}. To effectively leverage the augmented data for graph recommendation, we employ the CL paradigm as detailed in Section \ref{contrastive}.
\subsection{Data Augmentation via Majority Reranking}\label{LLM rerank}

This section presents a targeted data augmentation strategy aimed at enriching the training graph with high-quality synthetic interactions. Augmenting artificial interactions for high-degree users can introduce unnecessary cost and noisy interactions. Therefore, we apply the augmentation only to low-degree users, defined as those whose number of interactions falls below the 25th percentile of the user interaction distribution. This strategy ensures that edges are augmented only for users with insufficient interaction data. An ablation study analyzing the effect of this threshold is provided in Section~\ref{ablations}. 

For each user $u$ in the selected set, we aim to select $p$ new interactions that do not appear in the original dataset. 
While LLMs have shown promising capabilities in ranking and recommendation tasks, they are limited by context length and thus cannot consider all potential items and related users within a single prompt. Hence, we employ LightGCN as a conventional retrieval model to obtain an ordered shortlist of $K$ non-interacted candidate items, denoted as $\textbf{I}_{\text{can}} = \{i_k\}_{k=1}^K$, along with a set of similar users based on collaborative embeddings generated by the retrieval model, inspired by \citeauthor{wu2024coral,wei2024llmrec}. Based on this information, we employ an LLM to rerank the sequence $\textbf{I}_{\text{can}}$, producing a sorted list denoted as $\textbf{I}_{\text{sorted}}$. To perform the reranking, we design a generalizable few-shot prompting template consisting of four components: first, a task-specific instruction; second, the user’s watch history enriched with flattened metadata; third, a structured output format; and fourth, one or more few-shot examples drawn from similar users. Each example includes the user’s history along with a ground-truth, reranked list sorted by rating. This prompt is dataset-agnostic and can be easily adapted to new domains by reformatting the metadata. An example prompt for the movie dataset is provided in Figure \ref{fig1}, with further details available in Appendix~\ref{prompt}.

To improve the consistency of reranking outcomes and better leverage the LLM’s reasoning capabilities across multiple trials, we perform $N$ independent rerankings of the candidate list $\textbf{I}_{\text{can}}$. The resulting sequences are then aggregated using a simplified variant of RRF. For each item $i_k \in \textbf{I}_{\text{can}}$, we compute an aggregated score as follows:
\begin{equation}\label{rff}
g(i_k) = \sum_{n=1}^{N} \frac{1}{\varsigma^{(n)}(i_k) + 1}
\end{equation}
where $\varsigma^{(n)}(i_k)$ denotes the position of item $i_k$ in the $n$-th reranked sequence. We select the top $p$ items with the highest scores to construct new artificial interactions for the user $u$. Following this data augmentation strategy, we theoretically prove that for any bounded, strictly decreasing function $g(\cdot)$, increasing the number of independent re-rankings $N$ asymptotically reveals the LLM’s optimal parametric knowledge for the reranking task. The formal statement is provided in Theorem~\ref{thm:aggregation_general}, and a detailed proof can be found in Appendix~\ref{theorem}.

\begin{theorem}[Concentration for Score Aggregation over Random Permutations]
\label{thm:aggregation_general}
Let $\mathcal{I} = \{i_1, i_2, \ldots, i_K\}$ be a set of $K$ items, and let $\mathcal{T}_K$ denote all permutations of $\{1,2,\ldots,K\}$. Let $P$ be any probability distribution over $\mathcal{T}_K$, and let $\varsigma^{(1)}, \ldots, \varsigma^{(N)}$ be independent random permutations sampled from $P$. Let $g: \mathbb{R} \to \mathbb{R}$ be a bounded, strictly decreasing function. Suppose $A, B \in \mathbb{R}$ are such that $A \leq g(x) \leq B$ for all $x$. For each item $i \in \mathcal{I}$, define the aggregated score
\[
S(i) = \sum_{n=1}^N g\big(\varsigma^{(n)}(i)\big),
\]
where $\varsigma^{(n)}(i)$ is the rank assigned to item $i$ by permutation $\varsigma^{(n)}$. Let consider two distinct items $i_j, i_k \in \mathcal{I}$, and $ \mu = \mathbb{E}_{\varsigma \sim P}\left[ \varsigma(i_j) - \varsigma(i_k) \right]$. If $\mu > 0$, then
\[
\Pr\big( S(i_j) > S(i_k) \big) \leq \exp\left( -\frac{N\mu^2}{2(B-A)^2} \right).
\]
\end{theorem}

We instantiate $g(\cdot)$ as in Equation~\ref{rff} for its simplicity and prominence in prior work. More importantly, given a finite value of $K$, the $g(\cdot)$ function is strictly decreasing and bounded, satisfying the requirements of Theorem~\ref{thm:aggregation_general}. The corollary below specializes in the result for this choice, and its proof is presented in Appendix~\ref{theorem}.

\begin{corollary}[Reciprocal Rank Aggregation]
\label{cor:aggregation_rrf}
Let $g(x) = 1/(x+1)$ and define $S(i) = \sum_{n=1}^N 1/(\varsigma^{(n)}(i) + 1)$ for each $i \in \mathcal{I}$. For any $i_j, i_k \in \mathcal{I}$, let $\mu = \mathbb{E}_{\varsigma \sim P}\left[ \varsigma(i_j) - \varsigma(i_k) \right]$. If $\mu > 0$, then
\[
\Pr\big( S(i_j) > S(i_k) \big) \leq \exp\left( -\frac{N\mu^2}{2\left(1 - \frac{1}{K+1}\right)^2} \right).
\]
\end{corollary}

The probability distribution $P$ appearing in both the Theorem and Corollary is genuinely the LLM-induced distribution over permutations. Theorem~\ref{thm:aggregation_general} and Corollary~\ref{cor:aggregation_rrf} show that when item $i_k$ is expected to be ranked higher than item $i_j$, meaning it has a lower average rank and $\mu > 0$, the probability that $i_j$ receives a higher aggregated score than $i_k$ decreases exponentially with the number of permutations $N$. As $N \to \infty$, this error probability vanishes, indicating that score aggregation becomes increasingly reliable. Therefore, selecting the top $p$ items with the highest aggregated scores after enough rerankings recovers the optimal set of $p$ items under the LLM’s ideal parametric ranking.

We repeat the process for all targeted users, leading to $Q_{\alpha} \times |\mathcal{V}|\times p$ newly augmented edges. After data augmentation, we obtain an enriched graph $\mathcal{G}_{\text{aug}} = (\mathcal{V}, \mathcal{E}^+)$, where the original edge set $\mathcal{E}$ is augmented with newly generated interactions, thereby expanding $\mathcal{E}^+ = \mathcal{E} \cup \mathcal{E}_{\text{new}}$ while preserving the bipartite structure between users $\mathcal{U}$ and items $\mathcal{I}$.

\begin{table}[h]
\centering
\small
\setlength{\tabcolsep}{1.5pt}
\begin{tabular}{ccccc}
\hline
\textbf{Method} & \textbf{Virtual} & \textbf{Few-shot} & \textbf{Textual} & \textbf{Consistent} \\
 & \textbf{Interaction} & \textbf{Prompting} & \textbf{Embed.} & \textbf{Aug. Data} \\
\hline
VoteGCL & \checkmark & \checkmark & \texttimes & \checkmark \\
\hline
LLMRec  & \checkmark & \texttimes & \checkmark & \texttimes \\
\cite{wei2024llmrec} & & & & \\
\hline
RLMRec  & \texttimes & \texttimes & \checkmark & \texttimes \\
\cite{ren2024representation} & & & & \\
\hline
KAR  & \texttimes & \texttimes & \checkmark & \texttimes \\
\cite{xi2024towards} & & & & \\
\hline
\end{tabular}
\caption{Comparison of LLM augmented data methods}
\label{tab:llm-methods}
\end{table}

Unlike other LLM-augmented data approaches summarized in Table~\ref{tab:llm-methods}, VoteGCL does not incur the cost of generating textual embeddings for user/item profiling or adapting high-dimensional textual representations into recommendation models. Instead, VoteGCL focuses on generating high-quality interaction data and effectively leveraging it through a Graph CL paradigm, as detailed in Section~\ref{contrastive}. The key advantages of our LLM-augmented data are twofold: (1) a majority-vote reranking mechanism, which ensures more consistent and high-confidence user-item interactions, with detailed theoretical justification provided in Appendix~\ref{theorem}; and (2) a few-shot prompting strategy that incorporates collaborative signals from similar users, enabling better adaptation to user-specific preferences in the data augmentation task.
\subsection{Graph Contrastive Learning Paradigm}\label{contrastive}
We continue to employ LightGCN as our graph encoder due to its efficiency and effectiveness, though our framework is compatible with other graph backbones (e.g., NGCF \cite{wang2019neural}, XSimGCL \cite{yu2023xsimgcl}). Inspired by SimCLR \cite{chen2020simple}, we employ a CL paradigm between two representation views derived from the original and augmented data. This approach mitigates the distributional shift caused by potential noise introduced during graph augmentation. We share a common embedding matrix $\mathbf{E}^{(0)}$ to encode both the original graph $\mathcal{G}$ and its augmented version $\mathcal{G}_{\text{aug}}$. During each training epoch, node representations are generated as follows:
\[
\mathbf{E}^{(l)} = H(\mathbf{E}^{(l-1)}, \mathcal{G}_{\text{aug}}), \quad \mathbf{E}_{\text{org}}^{(l)} = H(\mathbf{E}^{(l-1)}, \mathcal{G})
\]
These embeddings serve as two distinct views of representation learning for contrastive supervision. The contrastive loss is then applied by treating embeddings of the same node across different views as positive pairs, and all others as negative:
\begin{equation}
    \mathcal{L}_{\text{cl}} = \sum_{i \in \mathcal{B}} -\log \frac{\exp \left(\mathbf{z}_i^{\top} \mathbf{z}_{i, \text{org}} / \tau\right)}{\sum_{j \in \mathcal{B}} \exp \left(\mathbf{z}_i^{\top} \mathbf{z}_{j, \text{org}} / \tau\right)},
\end{equation}
Here, $i, j$ represent user (item) nodes within a sampled mini-batch $\mathcal{B}$, and $\mathbf{z}, \mathbf{z}_{\text{org}}$ are the $L_2$-normalized $d$-dimensional representations obtained from the augmented graph $\mathcal{G}_{\text{aug}}$ and original graph $\mathcal{G}$. Note that, $\mathbf{z}_i=\frac{\mathbf{e}_i}{\left\|\mathbf{e}_i\right\|_2}$. By replacing $\tilde{A}$ with the adjacency matrices of graph $\mathcal{G}_{\text{aug}}$, $\mathbf{z}$ can be learned via Equation \ref{gcn}. The temperature parameter $\tau > 0$ (e.g., $\tau=0.2$) controls the sharpness of the similarity distribution. The CL loss can also be expressed as:
\begin{equation}
    \mathcal{L}_{cl}=\sum_{i \in \mathcal{B}}\left(-\mathbf{z}_i^{\top} \mathbf{z}_{i, \text{org}} / \tau+\log(\sum_{j \in \mathcal{B} /\{i\}} \exp (\mathbf{z}_i^{\top} \mathbf{z}_{j, \text{org}} / \tau))\right)
\end{equation}
Thus, $\mathcal{L}_{\text{cl}}$ enhances uniformity in the embedding space through its first component and promotes diversity through the second component, thereby improving model accuracy and mitigating popularity bias \cite{zimmermann2021contrastive, wang2020understanding}. The final node embeddings used for recommendation are obtained from the last propagation layer of the augmented graph. Specifically, we adopt $\mathbf{E}^{(L)}$, where $L$ corresponds to the total number of propagation layers, as the learned representation from the augmented graph $\mathcal{G}_{\text{aug}}$. This embedding captures the high-order structural information introduced by the augmentation and is subsequently employed to compute the BPR loss for optimizing recommendation performance. Finally, the overall training objective combines the recommendation loss with the CL loss:
\begin{equation}
    \mathcal{L}_{\text{main}} = \mathcal{L}_{\text{BPR}} + \lambda \mathcal{L}_{\text{cl}}
\end{equation}
where $\lambda\in(0,1)$ is a hyperparameter to control the weight of CL loss. The main advantages of this CL paradigm are that it leverages LLMs purely for knowledge-guided augmentation without embedding integration, thereby avoiding high-dimensional representations and complex fusion with collaborative signals. Additionally, VoteGCL shares the same learnable embedding matrix for both the original and augmented graphs, enabling more efficient training with reasonable model complexity. A detailed analysis of model complexity is provided in Appendix~\ref{complex}.
\section{Experiments}
\subsection{Experimental Setup}  
\textbf{Datasets.} We evaluate VoteGCL on four real-world datasets from diverse domains: Amazon (Scientific and Book) \citep{hou2024bridging}, MovieLens \cite{harper2015movielens}, Netflix \cite{wei2024llmrec}, and Yelp2018 \cite{wang2019neural}. Dataset statistics after preprocessing are summarized in Table~\ref{tab:dataset-stats}. By default, we select the top-1 result from majority-vote LLM reranking as virtual edges for data augmentation.

\textbf{Metrics.} We use NDCG@K and Recall@K to evaluate alignment between recommended items and user preferences, and APLT@K~\cite{abdollahpouri2019managing} to measure the average percentage of long-tail items among top-K recommendations as an indicator of popularity bias. Following the leave-one-out strategy, we use the latest interaction for testing, the second latest for validation, and the remaining for training. LightGCN generates item candidates for LLM reranking using \textit{gpt-3.5-turbo} \cite{achiam2023gpt} as prior work. 

\textbf{Baselines.} We compare VoteGCL against three categories: (1) base graph-based recommendation systems including LightGCN~\cite{he2020lightgcn}, NGCF~\cite{wang2019neural}, and MF-BPR~\cite{koren2009matrix}; (2) traditional augmented graph recommendation methods such as SGL~\cite{zhou2023selfcf}, GraphDA~\cite{fan2023graphda}, GraphAug~\cite{zhang2024graph}, and XSimGCL~\cite{yu2022graph}; and (3) LLM-enhanced recommendation systems including KAR~\cite{xi2024towards}, RLMRec~\cite{ren2024representation}, and LLMRec~\cite{wei2024llmrec}. We provide details of baselines in Appendix \ref{experiment}.
\subsection{Main Results}

\begin{table}
\centering
\small
\setlength{\tabcolsep}{2pt}  
\scalebox{0.7}{\begin{tabular}{ll|ccccccc|ccc|c}
\toprule
\textbf{Dataset} & \textbf{Metric} & \makecell{\textbf{MF-}\\\textbf{BPR}} & \makecell{\textbf{Light}\\\textbf{GCN}}& \textbf{NGCF} & \makecell{\textbf{Graph}\\\textbf{Aug}} & \makecell{\textbf{Graph}\\\textbf{DA}} & \textbf{SGL} & \textbf{XSimGCL} & \makecell{\textbf{NGCF+}\\\textbf{VoteGCL}} & \makecell{\textbf{XSimGCL+}\\\textbf{VoteGCL}} & \makecell{\textbf{LightGCN+}\\\textbf{VoteGCL}} & \textbf{Improv.}\\
\midrule
\multirow{3}{*}{\makecell{Movielen\\100K}} & R@10 & 0.0750 & 0.0766 & 0.0681 & 0.0689 & 0.0528 & 0.0818 & \underline{0.0852} & $0.0691^{0.3\%}$ & $0.0852^{0.0\%}$ & $\textbf{0.0886}^{15.7\%}$ & 4.0\%\\
& N@10 & 0.0335 & 0.0334 & 0.0308 & 0.0334 & 0.0251 & \underline{0.0407} & 0.0398 & $0.0372^{8.1\%}$ & $0.0411^{3.3\%}$ & $\textbf{0.0494}^{47.9\%}$ & 21.4\%\\
& A@10 & 0.0042 & 0.0097 & 0.0019 & 0.0016 & 0.0463 & 0.0227 & \underline{0.0583} & $0.0166^{733\%}$ & $\textbf{0.1852}^{217\%}$ & $0.0617^{536\%}$ & 217.7\%\\
\cmidrule(lr){1-13}
\multirow{3}{*}{\makecell{Movielen\\1M}} & R@10 & 0.0712 & 0.0817 & 0.0628 & 0.0643 & 0.0799 & 0.0812 & \underline{0.0857} & $0.0659^{4.9\%}$ & $0.0901^{5.1\%}$ & $\textbf{0.0928}^{13.6\%}$ & 8.3\%\\
& N@10 & 0.0365 & 0.0403 & 0.0311 & 0.0268 & 0.0385 & 0.0402 & \underline{0.0416} & $0.0375^{20.6\%}$ & $0.0428^{2.9\%}$ & $\textbf{0.0462}^{14.6\%}$ & 11.1\%\\
& A@10 & 0.0015 & 0.0027 & 0.0011 & 0.0056 & \underline{0.0104} & 0.0028 & 0.0098 & $0.0061^{454\%}$ & $0.0123^{25.5\%}$ & $\textbf{0.0145}^{437\%}$ & 39.4\%\\
\cmidrule(lr){1-13}
\multirow{3}{*}{Yelp2018} & R@10 & 0.0304 & 0.0425 & 0.0375 & 0.0352 & 0.0422 & \underline{0.0446} & 0.0435 & $0.0401^{6.9\%}$ & $0.0442^{1.6\%}$ & $\textbf{0.0487}^{14.6\%}$ & 8.8\%\\
& N@10 & 0.0147 & 0.0214 & \underline{0.0229} & 0.0125 & 0.0230 & 0.0223 & 0.0221 & $0.0201^{0.0\%}$ & $0.0234^{10.9\%}$ & $\textbf{0.0247}^{15.4\%}$ & 7.9\%\\
& A@10 & 0.0075 & 0.0038 & 0.0014 & 0.0012 & 0.0032 & \underline{0.0096} & 0.0068 & $0.0018^{28.6\%}$ & $0.0210^{208\%}$ & $\textbf{0.0349}^{818\%}$ & 263.5\%\\
\cmidrule(lr){1-13}
\multirow{3}{*}{\makecell{Amazon\\Scientific}} & R@10 & 0.1173 & 0.1251 & 0.0941 & 0.0986 & 0.1204 & \underline{0.1297} & 0.1291 & $0.1317^{39.9\%}$ & $0.1335^{3.4\%}$ & $\textbf{0.1385}^{10.7\%}$ & 6.8\%\\
& N@10 & 0.0850 & 0.0806 & 0.0648 & 0.0646 & 0.0758 & \underline{0.0911} & 0.0896 & $\textbf{0.1008}^{48.0\%}$ & $0.0936^{4.5\%}$ & $0.0993^{23.2\%}$ & 10.6\%\\
& A@10 & \underline{0.2563} & 0.2015 & 0.1839 & 0.0413 & 0.2034 & 0.2504 & 0.1779 & $0.2542^{38.2\%}$ & $0.1975^{11.0\%}$ & $\textbf{0.2795}^{38.7\%}$ & 9.1\%\\
\bottomrule
\end{tabular}}
\caption{Performance Comparison of Vanilla Graph Recommendation Methods on Multiple Datasets Before and After Data Augmentation. Each value in augmented data is shown as $\text{value}^\text{improvement}$, where the superscript indicates the percentage improvement over the baseline. Underlined values highlight the best baseline performance, while bold values highlight the best performance after data augmentation. "Improv" refers to the absolute improvement between these two.}
\label{tab:performance-comparison}
\end{table}

\begin{figure*}[t]  
\centering
\includegraphics[width=\textwidth]{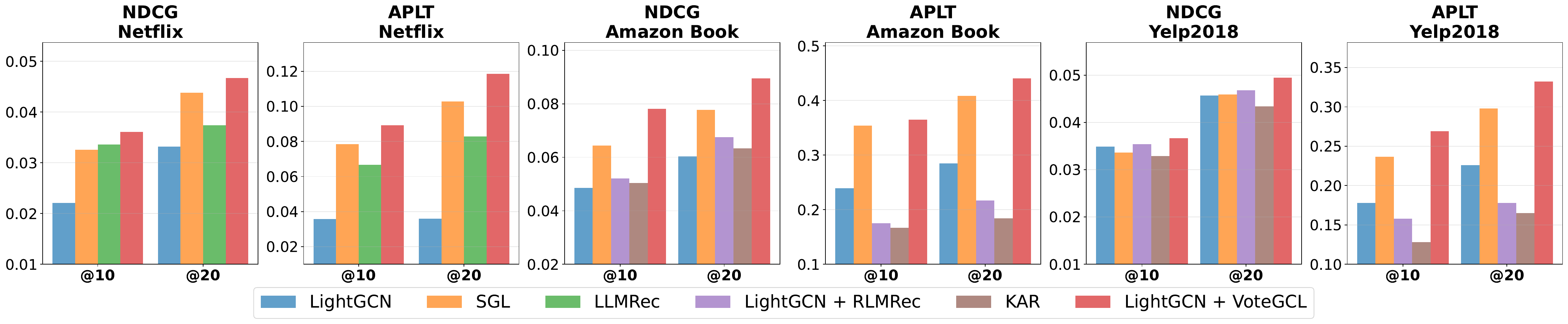}
\caption{Performance of LLM-enhanced graph recommenders on three datasets under Top K settings (K = 10 and 20), using dataset-specific models and leveraging available LLM-augmented data for fair comparison.}
\label{fig:llm_performance}
\end{figure*}
\subsubsection{Traditional graph recommendation.}
Table \ref{tab:performance-comparison} demonstrates that the proposed VoteGCL framework exhibits strong model-agnostic capabilities, consistently improving performance across diverse graph-based recommendation architectures, including LightGCN, XSimGCL, and NGCF. VoteGCL achieves substantial improvements in recommendation accuracy, with LightGCN+VoteGCL showing the most significant gains of $13.6\%$ on MovieLens 1M, and $14.6\%$ on Yelp2018 for R@10 metrics. Notably, the framework demonstrates exceptional scalability, delivering consistent improvements across both large-scale datasets (Yelp2018 and MovieLens 1M) and smaller datasets (Amazon Scientific and MovieLens 100K), with LightGCN+VoteGCL achieving particularly strong performance gains of $43.6\%$ on MovieLens 100K and $16.4\%$ on Amazon Scientific for NDCG@10, demonstrating superior adaptation to user preferences in reranking scenarios. More importantly, the framework effectively addresses popularity bias by dramatically enhancing long-tail item recommendations, as evidenced by APLT improvements ranging from $25.5\%$ to $818\%$ across different dataset-model combinations. These results demonstrate that VoteGCL successfully mitigates popularity bias while maintaining recommendation accuracy, promoting diversity by enabling models to learn better representations for underrepresented items regardless of dataset scale. The consistent performance gains across four distinct datasets spanning different domains and varying sizes validate the robust generalizability and universal applicability of VoteGCL.
\subsubsection{LLM-enhanced graph recommendation system.}
As shown in Figure~\ref{fig:llm_performance}, LightGCN + VoteGCL consistently outperforms all other methods across three benchmark datasets (Netflix, Amazon Book, and Yelp2018) under both @10 and @20 evaluation settings, achieving superior performance in terms of both accuracy and diversity. This includes outperforming traditional baselines (LightGCN, SGL) as well as LLM-enhanced models. While LLMRec, RLMRec, and KAR demonstrate notable improvements in ranking metrics such as NDCG, particularly when textual embeddings from LLMs are incorporated into the recommendation backbone, they also introduce a higher degree of popularity bias. This drawback is reflected in their lower APLT scores compared to SGL, which applies conventional data augmentation without relying on LLM-generated semantics. In contrast, VoteGCL alleviates this issue by generating high-confidence synthetic interactions through LLM prompting and integrating them effectively via graph contrastive learning, resulting in a more balanced trade-off between recommendation accuracy and fairness.
\begin{figure*}[t]  
\centering
\includegraphics[width=\textwidth]{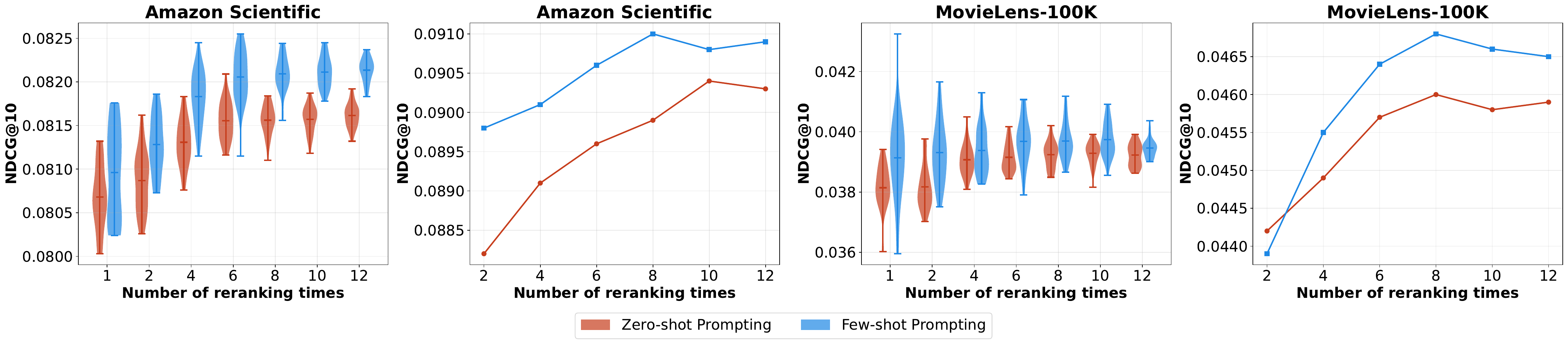}
\caption{Comparison of reranking performance between zero-shot and few-shot prompting across different numbers of majority-voting rounds. Results are averaged over 30 runs on the Amazon Scientific and MovieLens-100K datasets.}
\label{fig:comparison_zeroshot_fewshot}
\end{figure*}

\subsection{Ablation Studies}\label{ablations}

\begin{table}[t]
\centering
\small
\setlength{\tabcolsep}{2pt}
\begin{tabular}{lccccc}
\hline
\multirow{2}{*}{Dataset} & \multirow{2}{*}{Metric} & \multicolumn{4}{c}{Quantile of users to be augmented} \\
& & $25\%$ & $50\%$ & $75\%$ & $100\%$ \\
\hline
\multirow{2}{*}{MovieLens-100K} & N@10 & 0.0494 & 0.0482 & 0.0501 & \textbf{0.0503} \\
& A@10 & \textbf{0.0617} & 0.0562 & 0.0574 & 0.0596 \\
\hline
\multirow{2}{*}{Amazon Scientific} & N@10 & \textbf{0.0993} & 0.0983 & 0.0989 & 0.0991 \\
& A@10 & 0.2795 & \textbf{0.2939} & 0.2893 & 0.2792 \\
\hline
\multirow{2}{*}{Netflix} & N@10 & 0.0351 & 0.0336 & 0.0349 & \textbf{0.0352} \\
& A@10 & 0.0892 & \textbf{0.1089} & 0.0995 & 0.0919 \\
\hline
\end{tabular}
\caption{Performance under different ratios of users selected for augmentation}
\label{tab:virtual_edges}
\end{table}
\subsubsection{Number of Virtual Edges for Augmentation.}
Table~\ref{tab:virtual_edges} examines the impact of augmenting different user quantiles on VoteGCL's performance across three datasets. The results show that expanding augmentation to a larger portion of users does not consistently improve performance in either NDCG@10 or APLT@10. In some cases, performance even declines, suggesting that indiscriminate augmentation can introduce noise that offsets the benefits of additional interactions. These findings highlight the importance of targeted augmentation, specifically, focusing on the bottom $25\%$ of users with the fewest interactions. This approach improves generalization for cold-start users while maintaining overall model robustness and computational efficiency. Alternatively, the trade-off between recommendation accuracy and diversity can be more effectively managed by tuning the contrastive loss weight $\lambda$, allowing finer control over model behavior without increasing augmentation overhead.
\subsubsection{Zero-shot vs Few-shot Prompting LLM with Different Number of Majority-voting.}
Figure~\ref{fig:comparison_zeroshot_fewshot} illustrates the impact of the prompting strategy and the number of reranking rounds on the NDCG@10 performance of LLM-based item reranking. Across both datasets, few-shot prompting consistently outperforms zero-shot prompting, highlighting its effectiveness in capturing user preferences through collaborative signals. As the number of majority-voting rounds increases, performance improves and becomes more stable, particularly from 8 rounds onward, for both prompting strategies. Notably, the few-shot setting yields higher mean scores across 30 trials, indicating greater reliability and robustness in reranking outcomes. These findings suggest that effectively leveraging the LLM’s parametric knowledge via few-shot prompting leads to higher-quality augmented data and ultimately enhances the performance of VoteGCL. The performance at 6 voting rounds is competitive, making it a practical and cost-effective choice in resource-constrained scenarios. From these observations, we recommend combining few-shot prompting with a sufficiently large number of majority-voting rounds (e.g., $\geq 6$) to fully realize the benefits of LLM-based reranking in data augmentation.
\subsubsection{Controling Hyperparam $\lambda$ Contrastive Weight.}
\begin{figure}[t]  
\centering
\includegraphics[width=0.8\columnwidth]{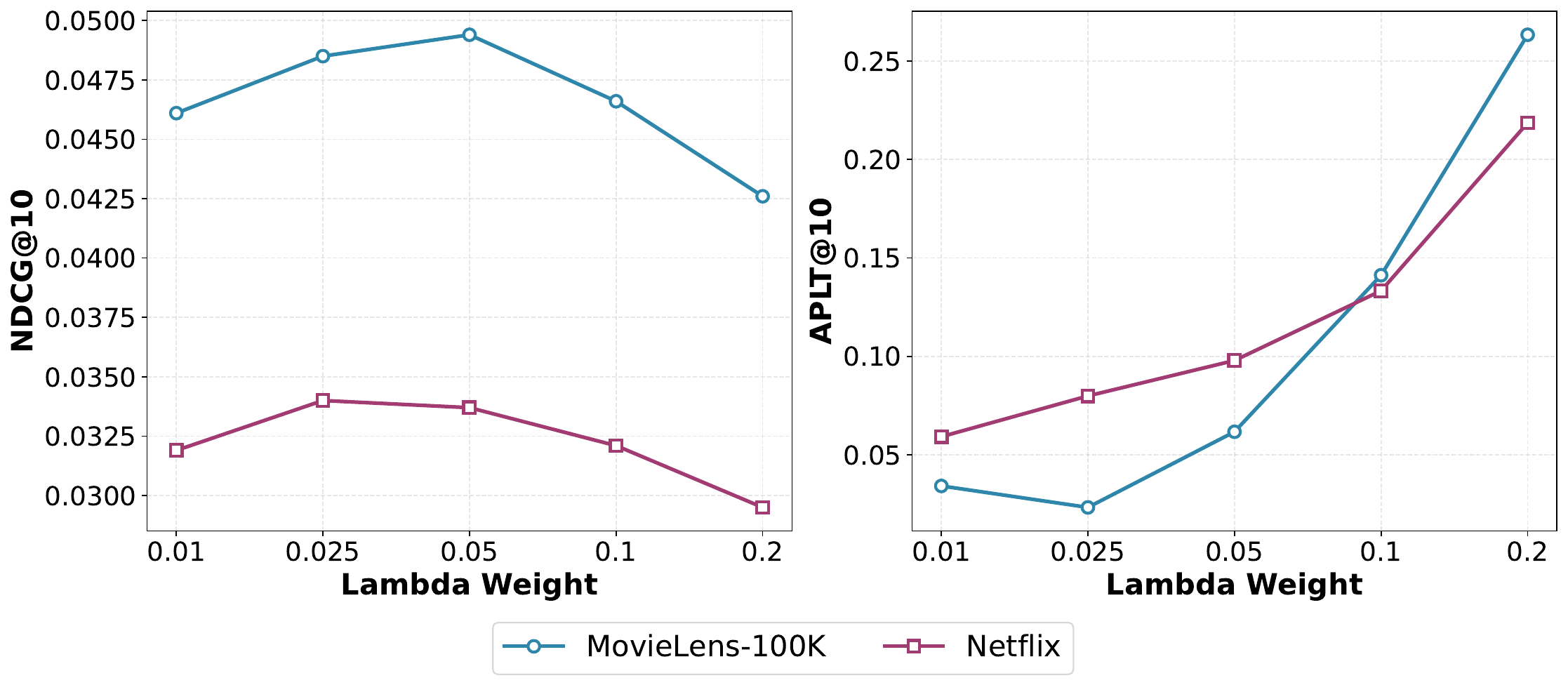}
\caption{Performance of VoteGCL under different contrastive weights}
\label{fig:lambda_contrastive}
\end{figure}
As illustrated in Figure~\ref{fig:lambda_contrastive}, the contrastive loss weight $\lambda \in (0,1)$ plays a critical role in balancing recommendation accuracy and popularity bias mitigation. When $\lambda$ increases, the APLT@10 metric improves significantly, indicating greater inclusion of long-tail items and enhanced diversity in recommendations. This effect is consistent across both MovieLens-100K and Netflix datasets. However, this gain in diversity comes with a trade-off: NDCG@10 gradually declines as $\lambda$ grows, reflecting a reduction in ranking precision. The results suggest that overly high values of $\lambda$ (e.g., 0.2) may overly prioritize long-tail exploration at the expense of accuracy. Based on these observations, we recommend tuning $\lambda$ within the range of 0.01 to 0.1 to strike an effective balance between accuracy and fairness, depending on the characteristics of the target dataset.
\subsubsection{Study of VoteGCL.}
\begin{table}[t]
\centering
\small
\setlength{\tabcolsep}{3pt}
\begin{tabular}{lcccccc}
\hline
\multirow{2}{*}{Methods} & \multicolumn{2}{c}{MovieLens-100K} & \multicolumn{2}{c}{Netflix} & \multicolumn{2}{c}{Amazon Book} \\
& N@10 & A@10 & N@10 & A@10 & N@10 & A@10 \\
\hline
VoteGCL & \textbf{0.0494} & \textbf{0.0617} & \textbf{0.0351} & \textbf{0.0892} & \textbf{0.0781} & \textbf{0.3654} \\
LightGCN & 0.0344 & 0.0097 & 0.0221 & \underline{0.0358} & 0.0486 & 0.2390 \\
w/o GCL & \underline{0.0382} & \underline{0.0102} & \underline{0.0245} & 0.0341 & \underline{0.0502} & \underline{0.2511} \\
w/o LLM augment & 0.0201 & 0.0068 & 0.0115 & 0.0126 & 0.0274 & 0.1243 \\
\hline
\end{tabular}
\caption{Ablation Study on VoteGCL Components}
\label{tab:ablation}
\end{table}
Table~\ref{tab:ablation} presents an ablation study evaluating the impact of key components in VoteGCL across three datasets. The full VoteGCL model consistently outperforms all baselines in both NDCG@10 and APLT@10. When the graph CL module is removed ("w/o GCL"), performance drops across all datasets, though it still outperforms LightGCN in certain metrics. This suggests that while the LLM-based augmentation is beneficial, without GCL, the model becomes more sensitive to noise from the synthetic data. Additionally, removing the LLM reranking with majority voting ("w/o LLM augment.") and selecting top-1 candidates without using LLM to vote leads to a significant performance collapse. This highlights the importance of majority voting in stabilizing the augmentation and leveraging textual descriptions modalities, making both GCL and LLM-based reranking essential for the overall robustness and effectiveness of VoteGCL.
\subsubsection{Additional Experimental Results.}
Additional ablation studies are provided in Appendix~\ref{experiment}. These include evaluations of different retrieval models for generating candidate items for LLM reranking, assessments of alternative LLMs for data augmentation (e.g., Qwen2.5~\cite{qwen2.5}), and an analysis of how varying the number of top-K items in the majority-voting process affects performance. Furthermore, we compare the cost of generating augmented data using VoteGCL in terms of price per accuracy gain and total generation cost, demonstrating that our approach provides a more cost-efficient alternative to other LLM-based methods.
\section{Conclusion}
We propose \textbf{VoteGCL}, a model-agnostic data augmentation framework that tackles inconsistency and distributional mismatch in LLM-enhanced recommendation. VoteGCL uses LLM-based reranking with majority voting to generate reliable synthetic interactions and applies contrastive learning to align original and augmented graphs without altering model architecture. Experiments across benchmarks show VoteGCL improves recommendation accuracy and reduces popularity bias, outperforming traditional and LLM-based data augmentation methods for a graph-based recommender.

\textbf{Limitations and scope.} VoteGCL reranks a pre-filtered candidate pool produced by a base retriever, which imposes a ceiling: items never retrieved cannot be recovered by the LLM, regardless of their true relevance. Consequently, VoteGCL is best viewed as a reliability amplifier on top of a retrieval stage rather than a mechanism for discovering items outside the retriever's reach. Two observations partially offset this limitation. First, Table \ref{tab:retrieval_comparison} shows that VoteGCL's gains transfer across retrievers and are larger when stronger retrievers (e.g., XSimGCL, SGL) are used, meaning the ceiling rises with retrieval quality. Second, because we target low-degree users and evaluate long-tail exposure via APLT, the augmented interactions consistently surface under-represented items that are within the retriever's top-K but would otherwise be ignored by collaborative scoring. Extending VoteGCL with a generative candidates proposal is a natural next step and would directly address the ceiling effect.

\bibliography{sn-bibliography.bib}

\newpage

\appendix
\section{More experimental results}\label{experiment}
\subsection{Experiment Setting}
We conduct experiments on several real-world datasets, with dataset statistics summarized in Table~\ref{tab:dataset-stats}. All models are trained for 100 epochs using 256-dimensional embeddings and a learning rate of $1\text{e}^{-3}$. Baseline hyperparameters are adopted from their original implementations. We employ LightGCN as the retrieval model to generate item candidates for LLM reranking. For augmentation, we follow prior work and use \textit{gpt-3.5-turbo}~\cite{achiam2023gpt}. Each graph-based recommendation model is configured with two layers. For VoteGCL, we evaluate contrastive weights in the range $(0.05, 0.1)$, with values specifically tuned for each dataset. All experiments are executed on a single NVIDIA A5000 GPU.
\begin{table}[h]
\setlength{\tabcolsep}{4pt}
\centering
\small  
\begin{tabular}{lrrrr}
\hline
\textbf{Datasets} & \textbf{Users} & \textbf{Items} & \textbf{Inters.} & \textbf{Density}  \\
\hline
Amazon Scientific & 10,995 & 5,314 & 76,091 & $1.3\times 10^{-3}$  \\
Amazon Book & 11,000 & 9,332 & 212,484 & $2.1\times 10^{-4}$ \\
Movielen-100K & 610 & 9,701 & 100,836 & $1.7\times 10^{-2}$ \\
Movielen-1M & 6000 & 25,317 & 914,108 & $6.2\times 10^{-3}$ \\
Yelp2018 & 11,091 & 86,646 & 544,672 & $5.3\times 10^{-4}$ \\
Netflix & 13,187 & 17,336 & 68,933 & $3.2\times 10^{-4}$ \\
\hline
\end{tabular}
\caption{Dataset statistics}
\label{tab:dataset-stats}
\end{table}

\textbf{Baselines.} We compare VoteGCL against three categories of baseline models: 
\begin{itemize}
    \item \textbf{Vanilla graph-based recommendation baselines.} LightGCN~\cite{he2020lightgcn} simplifies GCNs by using only neighborhood aggregation. NGCF~\cite{wang2019neural} propagates embeddings over a bipartite graph to capture collaborative signals. MF-BPR~\cite{koren2009matrix} applies matrix factorization with BPR loss to model user-item interactions.
    \item \textbf{Traditional augmented graph recommendation baselines.} SGL~\cite{zhou2023selfcf} applies contrastive learning to LightGCN using view augmentations to reduce bias and noise. GraphDA~\cite{fan2023graphda} refines graph structure with user-user/item-item correlations and top-K sampling. GraphAug~\cite{zhang2024graph} employs a graph information bottleneck for denoised, adaptive self-supervision. XSimGCL~\cite{yu2022graph} replaces graph augmentations with noise-based perturbations for simplified and effective contrastive learning.
    \item \textbf{LLM-enhanced recommendation baselines.} KAR~\cite{xi2024towards} augments recommendations with LLM-extracted reasoning and factual knowledge via a hybrid-expert adaptor. RLMRec~\cite{ren2024representation} aligns LLM-based semantic representations with collaborative signals through cross-view learning. LLMRec~\cite{wei2024llmrec} enriches user-item graphs using LLM-guided edge reinforcement, item attribute enhancement, and user profiling with denoising.
\end{itemize}

\textbf{Detail Metrics.} We provide definitions of all evaluation metrics used in our experiments as follows:
\begin{itemize}
 \item Average Percentage of Long Tail Items (APLT): measures the average percentage of long tail items in the recommended lists
$$
\text{APLT@K}=\frac{1}{\left|U_t\right|} \sum_{u \in U_t} \frac{\left|\left\{i, i \in\left(L_u@K \cap \Gamma\right)\right\}\right|}{\left|L_u@K\right|}
$$
Where item $i$ has been rated in the training set. $L_u@K$ is the recommended list of top $K$ relevant items for user $u$ and $\left|U_t\right|$ is the number of users in the test set, $\Gamma$ represents for long tail item set.
    \item Recall@K measures the proportion of relevant items that appear in the top-$K$ recommended list for each user:
$$
\text{Recall@K} = \frac{1}{|U_t|} \sum_{u \in U_t} \frac{|L_u@K \cap R_u|}{|R_u|}
$$
Where $R_u$ is the set of ground-truth relevant items for user $u$.
    \item NDCG@K evaluates the ranking quality of the top-$K$ recommended items, emphasizing higher ranks:
$$
\text{NDCG@K} = \frac{1}{|U_t|} \sum_{u \in U_t} \frac{\text{DCG@K}_u}{\text{IDCG@K}_u},
$$
$$
\text{DCG@K}_u = \sum_{i=1}^K \frac{\mathbb{I}(L_u[i] \in R_u)}{\log_2(i + 1)},
$$
$$
\text{IDCG@K}_u = \sum_{i=1}^{\min(K, |R_u|)} \frac{1}{\log_2(i + 1)}
$$
where $L_u[i]$ denotes the $i$-th item in the recommended list for user $u$, and $\mathbb{I}(\cdot)$ is the indicator function.
\end{itemize}
\subsection{Number of Augmented Interactions}
\begin{table}[t]
\centering
\small
\setlength{\tabcolsep}{4pt}
\begin{tabular}{lcccc}
\hline
\multirow{2}{*}{Dataset} & \multirow{2}{*}{Metric} & \multicolumn{3}{c}{Number of selected top-K} \\
& & 1 & 2 & 3 \\
\hline
\multirow{2}{*}{MovieLens-100K} & N@10 & \textbf{0.0494} & 0.0448 & 0.0439 \\
& A@10 & \textbf{0.0617} & 0.0593 & 0.0607 \\
\hline
\multirow{2}{*}{Amazon Scientific} & N@10 & \textbf{0.0993} & 0.0911 & 0.0887 \\
& A@10 & \textbf{0.2795} & 0.2041 & 0.1855 \\
\hline
\end{tabular}
\caption{Performance comparison across the number of majority-voting}
\label{tab:virtual_edges1}
\end{table}
Table \ref{tab:virtual_edges1} shows that selecting only the top-1 item after majority-voting LLM reranking consistently yields the best performance across both datasets, in terms of both NDCG@10 and APLT@10. As $K$ increases, the additional items introduced into the augmented data may contribute noise, potentially degrading recommendation quality. This effect is more pronounced in the Amazon Scientific dataset, where both metrics decline significantly when $K>1$. We therefore recommend fixing $K=1$ and adjusting the contrastive weight to better navigate the accuracy-diversity trade-off.
\subsection{Different retrieval models}
\begin{table*}[t]
\centering
\small
\setlength{\tabcolsep}{4pt}
\begin{tabular}{lcccccc}
\hline
\multirow{2}{*}{Retrieval} & \multicolumn{3}{c}{MovieLens-100K} & \multicolumn{3}{c}{Amazon Scientific} \\
& Base N@10 & N@10 & A@10 & Base N@10 & N@10 & A@10 \\
\hline
XSimGCL & 0.0398 & 0.0495 & \textbf{0.0662} & 0.0896 & \textbf{0.0993} & 0.2795 \\
SGL & \textbf{0.0407} & \textbf{0.0512} & 0.0608 & 0.0911 & 0.0950 & \textbf{0.2850} \\
LightGCN & 0.0334 & 0.0494 & 0.0617 & 0.0853 & 0.0920 & 0.2720 \\
\hline
\end{tabular}
\caption{Performance comparison of different retrieval methods and VoteGCL on retrieval-augmented data. "Base N@10" denotes the top-10 recommendation performance of the retrieval models themselves, while "N@10" and "A@10" reflect the performance of VoteGCL using the corresponding retrieved candidates.}
\label{tab:retrieval_comparison}
\end{table*}
Table~\ref{tab:retrieval_comparison} presents a comparison of different retrieval models and their impact on the performance of the VoteGCL reranking framework across two datasets. The results show that stronger retrieval models, such as XSimGCL and SGL, generally yield higher quality candidate items, which in turn lead to improved reranking performance by VoteGCL, as reflected in higher N@10 and A@10 scores compared to when using LightGCN. For instance, on the MovieLens-100K dataset, using XSimGCL as the retrieval model enables VoteGCL to achieve an A@10 of 0.0662, outperforming the LightGCN-based pipeline. Similarly, in the Amazon Scientific dataset, SGL provides the best A@10 of 0.2850. These findings suggest that the choice of retrieval model has a significant influence on the overall effectiveness of the recommendation system. While we recommend deploying stronger retrieval models in industrial settings to maximize performance, LightGCN remains a practical option due to its simplicity, fast inference, and reasonably competitive accuracy.
\subsection{Different LLMs to augment data}
\begin{table*}[h]
\centering
\small
\setlength{\tabcolsep}{6pt}
\begin{tabular}{lcccccc}
\hline
\multirow{2}{*}{LLM} & \multicolumn{3}{c}{MovieLens-100K} & \multicolumn{3}{c}{Amazon Scientific} \\
& Base N@10 & N@10 & A@10 & Base N@10 & N@10 & A@10 \\
\hline
gpt-3.5-turbo & \textbf{0.0395} & \textbf{0.0494} & \textbf{0.0617} & \textbf{0.0821} & \textbf{0.0993} & \textbf{0.2795} \\
Qwen2.5-14B-Instruct & 0.0342 & 0.421 & 0.0598 & 0.0784 & 0.0926 & 0.2612 \\
\hline
\end{tabular}
\caption{Ranking performance comparison of different LLMs on two datasets, along with the performance of VoteGCL on the corresponding LLM-augmented data. "Base N@10" indicates the LLM's recommendation performances, while "N@10" and "A@10" represent VoteGCL performance with LLM-augmented inputs.}
\label{tab:llm_performance}
\end{table*}
Table~\ref{tab:llm_performance} compares the ranking performance of different large language models (LLMs) and their impact on the VoteGCL reranking framework across two datasets. Overall, GPT-3.5-turbo consistently outperforms Qwen2.5-14B-Instruct across all metrics, likely due to its stronger parametric knowledge and more robust instruction-following capabilities. This advantage is reflected in both the base recommendation performance and the downstream gains achieved by VoteGCL using the LLM-augmented inputs. Nevertheless, Qwen2.5-14B-Instruct still delivers competitive results, making it a viable free alternative for scenarios where cost or deployment constraints prohibit the use of proprietary APIs like GPT. This highlights a practical trade-off between performance and accessibility when selecting LLMs for recommendation tasks.
\subsection{Cost Efficency}
\subsubsection{Cost and accuracy trade off.}
\begin{figure*}[t]
    \centering
    \includegraphics[width=\textwidth]{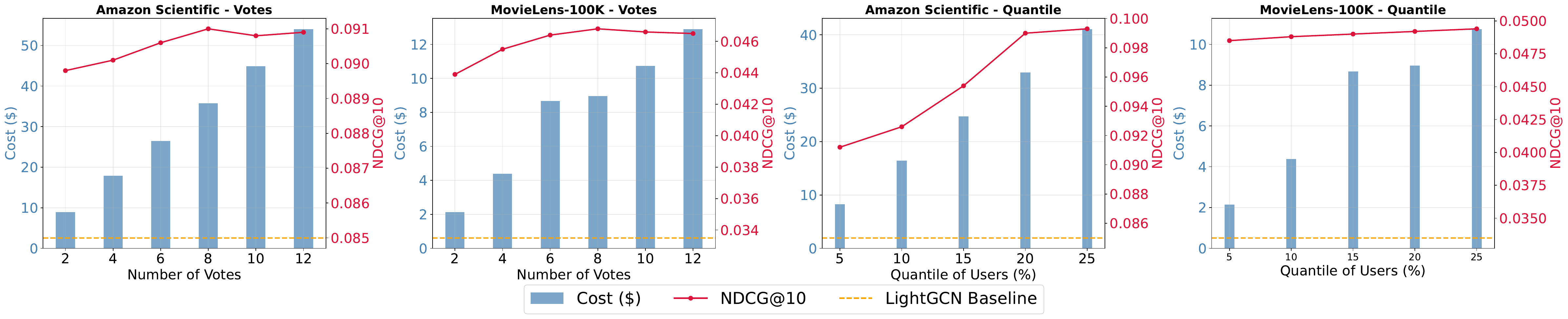}
    \caption{Cost and NDCG@10 performance trade-offs across two datasets (Amazon Scientific and MovieLens-100K). Left: Increasing the number of majority-vote LLM rerankings improves accuracy but raises inference cost. Right: Augmenting a larger quantile of user data yields accuracy gains with corresponding increases in cost.}
    \label{fig:cost_ndcg}
\end{figure*}
Figure \ref{fig:cost_ndcg} illustrates the cost-accuracy trade-offs associated with two complementary strategies for improving recommendation performance: increasing the number of LLM-based majority-vote rerankings and augmenting data for a subset of users. On the left, we observe that increasing the number of votes leads to consistent improvements in NDCG@10 across both datasets, albeit with rising inference cost. On the right, augmenting a greater proportion of users also yields notable accuracy gains with incremental cost. Notably, both evaluations show that substantial improvements in NDCG@10 can be achieved with relatively modest cost increases, for instance, from just 6 votes onward or augmenting only 10$\%$ of users. Based on these trends, we adopt a configuration using 8-way majority voting and data augmentation for 25\% of users, which provides a strong balance between accuracy and cost.
\subsubsection{Cost Comparison with Different Numbers of Candidates}
\begin{table*}[t]
\centering
\small
\setlength{\tabcolsep}{6pt}
\scalebox{0.9}{\begin{tabular}{lcccccc}
\hline
\multirow{2}{*}{\textbf{Number of candidates}} & \multicolumn{3}{c}{Amazon Scientific} & \multicolumn{3}{c}{MovieLens-100K} \\
& k=5 & k=10 & k=20 & k=5 & k=10 & k=20 \\
\hline
Cost (\$)& 35.73 & 41.05 & 48.02 & 8.07 & 8.97 & 10.14 \\
Recall@10 & 0.1378 & 0.1386 & \textbf{0.1437} & 0.0801 & \textbf{0.0886} & 0.0858 \\
NDCG@10 & 0.0984 & 0.0994 & \textbf{0.1043} & 0.0414 & \textbf{0.0494} & 0.0463 \\
Rate of Recall improvement to cost (\%)& \textbf{75.37} & 67.46 & 67.44 & 56.61 & \textbf{174.65} & 118.40 \\
Rate of NDCG improvement to cost (\%)& \textbf{96.86} & 87.92 & 89.09 & 296.76 & \textbf{534.07} & 380.74 \\
\hline
\end{tabular}}
\caption{Performance comparison across different numbers of retrieved candidates injected into the prompt. We report the total cost (in USD). The rates of improvement relative to cost are computed as the percentage gain over the LightGCN baseline divided by the corresponding cost}
\label{tab:performance_comparison}
\end{table*}
Table~\ref{tab:performance_comparison} presents the performance and cost trade-offs when varying the number of retrieved candidates injected into the prompt. Increasing the number of candidates retrieved from the LightGCN model enhances the likelihood of including relevant items, leading to improved recall and NDCG@10 scores. However, this comes at the cost of longer input prompts, resulting in increased token usage and inference cost. While the best performance is achieved at $k=20$, we observe that $k=10$ offers a more favorable balance between accuracy gains and cost, making it our default choice for downstream prompting.
\subsubsection{Cost Comparison with other methods.}
\begin{table}
\centering
\small
\setlength{\tabcolsep}{6pt}
\scalebox{0.75}{\begin{tabular}{llccccccc}
\hline
\textbf{Dataset} & \textbf{Method} & \textbf{Cost(\$)} & \textbf{N@10} & \textbf{A@10} & \textbf{Imp N}(\%) & \textbf{CIR N}(\%) & \textbf{Imp A}(\%) & \textbf{CIR A}(\%) \\
\hline
\multirow{2}{*}{Netflix} & Without reasoning & 9.47 & 0.0330 & 0.0639 & 49.32 & 5.21 & 134.36 & 14.19 \\
& With reasoning & 21.43 & 0.0351 & 0.0892 & 58.83 & 2.75 & 149.16 & 6.96 \\
\hline
LLMRec & & 21.14 & 0.0336 & 0.0667 & 52.04 & 2.46 & 86.31 & 4.08 \\
\hline
\end{tabular}}
\caption{Performance comparison on the Netflix dataset and the LLMRec baseline. We report the cost (in USD), NDCG@10 (N@10), and Accuracy@10 (A@10), along with the percentage improvements over the LightGCN baseline (Imp N and Imp A) and their corresponding cost-improvement ratios (CIR N and CIR A), calculated as improvement per dollar.}
\label{tab:netflix_llmrec}
\end{table}

Table~\ref{tab:netflix_llmrec} compares the performance of our method with and without reasoning prompts against the LLMRec \cite{wei2024llmrec} baseline on the Netflix dataset. By leveraging the reasoning capabilities of LLMs in our prompt design, we guide the model to produce a reranked list explicitly. When disabling reasoning in the prompt, the output is limited to a simple format (e.g., "A-B-C-...") without any natural language explanation, resulting in minimal output tokens and significantly reduced cost. Notably, even with this no-reasoning prompt, the method still achieves substantial improvements in NDCG@10 relative to cost, outperforming the LLMRec baseline in both effectiveness and cost-efficiency. Under the default reasoning setting, VoteGCL demonstrates strong overall performance while maintaining a favorable cost-performance trade-off compared to LLMRec.
\clearpage
\section{Model Complexity}\label{complex}
We compare the time complexity of VoteGCL with LightGCN and SGL-ED under a single-batch training setup, where in-batch negative sampling is adopted. Let $|\mathcal{E}|$ be the number of edges in the user-item bipartite graph, $d$ the embedding dimension, $B$ the batch size, $M$ the number of nodes in a batch, and $L$ the number of GCN layers.
\begin{table}[h]
\centering
\small
\setlength{\tabcolsep}{1pt}
\begin{tabular}{l|c|c|c}
\hline
\textbf{Component} & \textbf{LightGCN} & \textbf{SGL-ED} & \textbf{VoteGCL} \\
\hline
Adjacency & $O(2|\mathcal{E}|)$ & $O(2|\mathcal{E}|+4\rho|\mathcal{E}|)$ & $O(2(|\mathcal{E}|+|\mathcal{E}^+|))$ \\
Matrix & & & \\
\hline
Graph & $O(2|\mathcal{E}|Ld)$ & $O((2+4\rho)|\mathcal{E}|Ld)$ & $O(2Ld(|\mathcal{E}|+|\mathcal{E}^+|))$ \\
Encoding & & & \\
\hline
Prediction & $O(2Bd)$ & $O(2Bd)$ & $O(2Bd)$ \\
\hline
Contrast & - & $O(Bd+BMd)$ & $O(Bd+BMd)$ \\
\hline
\end{tabular}
\caption{The comparison of time complexity}
\label{tab:complexity}
\end{table}\\
Table \ref{tab:complexity} summarizes the time complexity comparison between LightGCN, SGL-ED, and our proposed VoteGCL. We observe that the main computational overhead of VoteGCL lies not in graph processing itself but in the data augmentation phase, where we apply Majority LLM Reranking to generate high-quality virtual interactions $\mathcal{E}^+$. This step occurs offline and is performed once before training, introducing no runtime burden during model optimization. In contrast, SGL-ED incurs significantly higher complexity during training due to explicit graph augmentations. It constructs two additional perturbed adjacency matrices, resulting in a higher cost for both adjacency preparation and graph encoding. Specifically, the encoding complexity of SGL-ED is roughly $(2 + 4\rho)$ times that of LightGCN, whereas VoteGCL only adds a modest $2|\mathcal{E}^+|Ld$ term. Furthermore, all models share the same prediction cost under BPR loss with in-batch negatives. For contrastive learning, VoteGCL and SGL-ED have similar complexity $\mathcal{O}(Bd + BMd)$, while LightGCN does not employ contrastive learning. Overall, VoteGCL achieves a favorable trade-off between performance and efficiency. By shifting the augmentation cost to a one-time LLM-based process, it maintains low training complexity comparable to LightGCN while benefiting from richer supervision akin to SGL-ED.
\section{Detail LLM Prompt}\label{prompt}
\begin{listing}[h]
\caption{Detail Prompt Template for Movie Reranking}
\label{lst:llm_prompt}
\begin{lstlisting}[basicstyle=\tiny\ttfamily, numbers=none]
You are a movie recommendation system. Given the user history in chronological order, recommend an item from the 
candidate pool with its index letter.

USER HISTORY:
{user_history}
The user history shows the movies the user has watched in chronological order. 
The first movie is the oldest one the user watched, and the last movie is the most recent one. 
Each entry follows this format: [Position]. [Movie Title] 
([Release Year]) [Genres] - Rating: [User Rating]
The user rating ranges from 1.0 to 5.0, with 5.0 being the highest level of enjoyment.

TOP 10 CANDIDATE MOVIE LIST:
{candidate_list} 
Each candidate movie is listed with an index letter 
(A, B, etc.) followed by the movie title, 
release year, and genres. 

REFERENCE EXAMPLE: A similar user had the following history: 
{similar_user_history}
This user rated the candidate product as:
{candidate_user_history}

Your task is to reorder these candidates based on the user's preferences, where the first movie should be the one 
the user would most likely enjoy, and subsequent movies represent decreasing levels of preference, and based on 
the reference example, too. Please analyze and summarize the user's preferences in paragraph form, and write it 
in <think></think> tags at the beginning of your response. Then, explain your reasoning for the ordering 
of the candidate movies and write it in <reasoning></reasoning> tags. Finally, provide your recommended ordering 
of ALL candidates movies as a hyphen-separated list of indices (e.g., A-B-C-D-E-F-G-H-I-J) and place it in 
<output></output> tags.

Make sure to include ALL 10 movie indices in your response. 
The first index should represent the movie you believe the user would enjoy most, with subsequent indices 
representing decreasing levels of preference.
\end{lstlisting}
\end{listing}

Template ~\ref{lst:llm_prompt} illustrates the prompt template we designed for the movie recommendation task, using the MovieLens dataset as a primary example. However, the template is designed to be generalizable and can be easily adapted to other datasets with minimal modification. One key design decision is the flattening of user history metadata into a structured natural language format: [Position]. [Movie Title] ([Release Year]) [Genres] - Rating: [User Rating]. This flattening not only makes the prompt more readable and coherent to the language model but also facilitates compatibility with diverse datasets by allowing simple reformatting of metadata fields. Furthermore, we cast the recommendation task as a reranking problem using a verbalized format. Rather than requiring the model to regenerate full movie titles, the prompt provides an indexed list of candidate movies and asks the model to output a ranked sequence of index letters. This design significantly reduces generation cost and latency, while maintaining high interpretability. To encourage transparency and structured reasoning, the prompt explicitly instructs the model to include preference analysis and justification using $<think>$ and $<reasoning>$ tags, followed by the final ordered list in an $<output>$ tag.

Additionally, the prompt incorporates collaborative filtering by injecting a reference example from a similar user. We identify the most similar user by computing cosine similarity between user embeddings, then include their history in the prompt. The last three rated items from a similar user serve as ground truth to further guide the model. This structured response format improves reproducibility and effectively leverages collaborative signals in recommendations. In practical deployments, the reasoning component can be made optional to reduce inference cost while preserving the core reranking functionality.
\section{Theoretical results}\label{theorem}
\begin{proof}[Proof of Theorem~\ref{thm:aggregation_general}]
    For convenience, we recall our notations in the main paper, followed by our formal proof. 

We consider a set of $K$ items retrieved by a retrieval system, denoted by $\mathcal{I} = \{i_1, i_2, \dots, i_K\}$. Let $\mathcal{T}_K$ denote the set of all permutations over $K$ elements:

$$
\mathcal{T}_K = \left\{ \sigma : \{1, 2, \ldots, K\} \to \{1, 2, \ldots, K\} \,\middle|\, \sigma \text{ is bijective} \right\}.
$$

Each permutation $\sigma \in \mathcal{T}_K$ represents a ranking assignment over the items in $\mathcal{I}$. Specifically, $\sigma(i_k)$ denotes the **\textbf{rank}** assigned to item $i_k$. In other words, lower values of $\sigma(i_k)$ indicate higher relevance or priority in the ranked list.

We define a random variable $\varsigma$ taking values in $\mathcal{T}_K$, with probability distribution

\[P : \mathcal{T}_K \to [0, 1], \text{ such that} \sum_{\sigma \in \mathcal{T}_K} P(\varsigma = \sigma) = 1.\]

In our work, all probabilities and distributions are conditioned on the prompt information, denoted as $\textbf{c}$, which includes user and item information. We focus on one user $u$, meaning that $\textbf{c}$ is fixed. Hence, we eliminate the appearance of $\textbf{c}$-dependence in our proof for better presentation.  Throughout this work, we use $\sigma$ to represent a permutation in $\mathcal{T}_K$, and write $P(\sigma)$ as shorthand for $P(\varsigma = \sigma)$, the probability assigned to permutation $\sigma$.

We draw $N$ independent permutations $\varsigma^{(1)}, \varsigma^{(2)}, \ldots, \varsigma^{(N)}$ from the distribution $P$. For each item $i_k \in \mathcal{I}$, where $k = 1, \ldots, K$, we define $\varsigma^{(n)}(i_k)$ as the rank assigned to item $i_k$ in the $n$-th permutation. Fix two items \(i_j\) and \(i_k\), with \(j, k \in \{1, \ldots, K\}\). Define
\[
S_j = \sum_{n=1}^{N} g\big(\varsigma^{(n)}(i_j)\big), \quad S_k = \sum_{n=1}^{N} g\big(\varsigma^{(n)}(i_k)\big),
\]
where $g: \mathbb{R} \rightarrow\mathbb{R}$ denotes any bounded strictly decreasing function. This assumption is natural and intuitive because items with a lower rank should have a higher score. 



Let us define the difference random variables:
\[
D_n = g\big(\varsigma^{(n)}(i_j)\big) - g\big(\varsigma^{(n)}(i_k)\big), \quad \text{for } n = 1, \ldots, N.
\]
Then we can write:
\[
S_j - S_k = \sum_{n=1}^{N} D_n.
\]

Because the function $g(\cdot)$ is bounded, there exist real numbers $A < B$ such that for all $x$, $A \leq g(x) \leq B$. Therefore, for any permutation $\varsigma$, the difference $D_n = g(\varsigma^{(n)}(i_j)) - g(\varsigma^{(n)}(i_k))$ satisfies $-(B-A) \leq D_n \leq B-A$.

Let $\mathbb{E}[D_n]$ be the expected score difference between $i_j$ and $i_k$ under the distribution $P$.

Applying Hoeffding's inequality: if $D_1, \ldots, D_N$ are independent with $D_n \in [a, b]$, then for any $t > 0$,
\[
\mathbb{P} \left( \sum_{n=1}^{N} (D_n - \mathbb{E}[D_n]) \ge t \right)
\le \exp\left( - \frac{2t^2}{N(b - a)^2} \right).
\]

In our setup, $a = -(B-A)$ and $b = B-A$, so $b-a = 2(B-A)$. Thus,
\begin{align*}
    \mathbb{P}(S_j > S_k)
&= \mathbb{P}\left( \sum_{n=1}^{N} D_n > 0 \right) \\
&= \mathbb{P}\left( \sum_{n=1}^{N} (D_n - \mathbb{E}[D_n]) > -N\mathbb{E}[D_n] \right).
\end{align*}

If $\mathbb{E}[D_n] < 0$, then $-N \mathbb{E}[D_n] > 0$ and Hoeffding's inequality gives

\begin{align*}
    \mathbb{P}(S_j > S_k)
&\leq \exp\left(-\frac{2(-N\mu)^2}{N[2(B-A)]^2}\right)\\
&= \exp\left( -\frac{N\mu^2}{2(B-A)^2} \right).
\end{align*}

Let us define $\mu = \mathbb{E}[\varsigma^{(n)}(i_j) - \varsigma^{(n)}(i_k)]$. Because $g(\cdot)$ is bounded and strictly decreasing, $\mathbb{E}[D_n] > 0$ is equivalent that $\mu > 0$
\end{proof}

This formally characterizes the concentration of the aggregated comparison as $N$ increases. The derived concentration result provides a rigorous guarantee for comparing the aggregated scores for pairs of items under repeated sampling from a stochastic permutation model. The bound demonstrates that, as long as the expected difference in scores \(\mu = \mathbb{E}[D_n]\) is strictly negative, the probability that the "worse" item \(i_j\) outperforms the "better" item \(i_k\) after \(N\) samples decays exponentially in \(N\). Notably, the exponent depends quadratically on the absolute value of the mean difference \(|\mu|\) and inversely on the squared range of the score differences \((B-A)^2\). Thus, for a reasonably well-separated pair of items and a bounded scoring function, even a modest number of samples \(N\) suffices to make the ordering highly reliable. 

Before proceeding with the proof for Corollary~\ref{cor:aggregation_rrf}, we provide its self-contained formal statement in Corollary~\ref{cor:aggregation_rrf_full}.

\begin{corollary}[Reciprocal Rank Aggregation]
\label{cor:aggregation_rrf_full}
Let $\mathcal{I} = \{i_1, i_2, \ldots, i_K\}$ be a set of $K$ items, and let $\mathcal{T}_K$ denote the set of all permutations over $\{1, 2, \ldots, K\}$. Let $P$ be any probability distribution over $\mathcal{T}_K$, and let $\varsigma^{(1)}, \ldots, \varsigma^{(N)}$ be independent random permutations drawn from $P$.

Define the score function $g(x) = 1/(x+1)$, and for each item $i \in \mathcal{I}$, define its aggregated score as
\[
S(i) = \sum_{n=1}^N \frac{1}{\varsigma^{(n)}(i) + 1},
\]
where $\varsigma^{(n)}(i)$ denotes the rank assigned to item $i$ by permutation $\varsigma^{(n)}$.

Fix two distinct items $i_j, i_k \in \mathcal{I}$, and define the expected rank gap $\mu = \mathbb{E}[\varsigma^{(n)}(i_j) - \varsigma^{(n)}(i_k)]$.
If $\mu > 0$, then the probability that item $i_j$ receives a higher aggregated score than item $i_k$ satisfies the concentration bound
\[
\Pr\big( S(i_j) > S(i_k) \big) \leq \exp\left( -\frac{N\mu^2}{2\left(1 - \frac{1}{K+1}\right)^2} \right).
\]
\end{corollary}

\begin{proof}[Proof of Corollary~\ref{cor:aggregation_rrf}]
For $x \in \{1, 2, \ldots, K\}$, $g(x)$ is maximized when $x = 1$, so $B = 1$. It is minimized when $x = K$, so $A = 1/(K+1)$. Therefore, for all $i_j, i_k \in \mathcal{I}$ and for all permutations, the values of $g$ satisfy $A \leq g(x) \leq B$, and $B - A = 1 - \frac{1}{K+1} = \frac{K}{K+1}$. Substituting $A$ and $B$ into Theorem~\ref{thm:aggregation_general} gives the stated bound.
\end{proof}

\end{document}